\documentclass[12pt]{article}
\usepackage[margin=1in]{geometry}
\usepackage[english]{babel}
\usepackage[utf8]{inputenc}
\usepackage[T1]{fontenc}
\usepackage{graphicx}
\graphicspath{{figures/}}
\usepackage{amssymb,amsmath,amsthm}
\usepackage{stmaryrd}
\usepackage{hyperref}
\usepackage{enumitem}
\setitemize{noitemsep,topsep=2pt,parsep=0pt,partopsep=0pt}
\setenumerate{noitemsep,topsep=2pt,parsep=0pt,partopsep=0pt}
\usepackage[table]{xcolor}
\usepackage{authblk}
\usepackage{tikz}
\usepackage{multicol,multirow}

\newcommand\tr{0.85}

\newtheorem{theorem}{Theorem}
\newtheorem{definition}[theorem]{Definition}

\newtheorem{lemma}[theorem]{Lemma}
\newtheorem{corollary}[theorem]{Corollary}

\newtheorem{remark}[theorem]{Remark}

\newcommand{\1}{{\normalfont \texttt{1}}}
\newcommand{\0}{{\normalfont \texttt{0}}}

\renewcommand{\O}{\mathcal{O}}
\newcommand{\Poly}{\mathsf{P}}
\newcommand{\NP}{\mathsf{NP}}
\newcommand{\SPoly}{\mathsf{\#P}}

\newcommand{\ID}[1]{G_{#1}}
\newcommand{\lab}{\mathrm{lab}}
\newcommand{\labminus}{\ominus}
\newcommand{\labplus}{\oplus}
\newcommand{\UD}[1]{\texttt{UD}(#1)}
\newcommand{\SUD}[1]{\#\UD{#1}}
\newcommand{\UDto}[4]{\UD{#1}_{#2\to#3}^{#4}}
\newcommand{\UDp}[3]{\UD{#1}_{#2}^{#3}}
\newcommand{\SUDp}[3]{\#\UDp{#1}{#2}{#3}}
\newcommand{\seqcomp}[2]{\mathcal{S}(#1,#2)}
\newcommand{\parcomp}[2]{\mathcal{P}(#1,#2)}
\newcommand{\freecomp}[2]{\mathcal{F}(#1,#2)}
\newcommand{\un}[1]{\bar{#1}}
\newcommand{\revneg}[1]{\tilde{#1}}

\newcommand{\decisionpb}[3]{\fbox{\parbox{0.9\textwidth}{{\bf #1}\\{\it Input:} #2\\{\it Question:} #3}}}

\newcommand{\countingpb}[3]{\fbox{\parbox{0.9\textwidth}{{\bf #1}\\{\it Input:} #2\\{\it Output:} #3}}}

\begin{document}

\title{\#P-completeness of counting update digraphs,\\
cacti, and a series-parallel decomposition method}
\author[1]{Camille No{\^u}s}
\author[1]{K{\'e}vin Perrot}
\author[1]{Sylvain Sen{\'e}}
\author[1]{Lucas Venturini}
\affil[1]{Universit{\'e} publique, France}
\date{}
\maketitle

\begin{abstract}
  Automata networks are a very general model of interacting entities, with
  applications to biological phenomena such as gene regulation. In many
  contexts, the order in which entities update their state is unknown,
  and the dynamics may be very sensitive to changes in this schedule of
  updates. Since the works of Aracena et. al, it is known that update
  digraphs are pertinent objects to study non-equivalent block-sequential
  update schedules. We prove that counting the number of equivalence classes,
  that is a tight upper bound on the synchronism sensitivity of a given
  network, is $\SPoly$-complete.
  The problem is nevertheless computable in quasi-quadratic time for oriented 
  cacti, and for oriented series-parallel graphs thanks to a decomposition method.
\end{abstract}

\section{Introduction}

Since their introduction by McCulloch and Pitts in the 1940s through
the well known formal neural networks~\cite{mcp43},
automata networks (ANs) are a general model of interacting entities
in finite state spaces. The field has important contributions to
computer science, with Kleene's finite state automata~\cite{k51},
linear shift registers~\cite{h59}
and linear networks~\cite{e59}.
At the end of the 1960s, Kauffman and Thomas (independently)
developed the use of ANs for the modeling of biological
phenomena such as gene regulation~\cite{k69,t73},
providing a fruitful theoretical framework~\cite{r69}.

ANs can be considered as a collection of local functions (one per component),
and influences among components may be represented as a so called {\em interaction digraph}.
In many applications the order of components update is a priori unknown,
and different schedules may greatly impact the dynamical properties of the system.
It is known since the works of Aracena {\em et al.} in~\cite{agms09} that {\em update digraphs}
(consisting of labeling the arcs of the interaction digraphs with $\labplus$ and $\labminus$)
capture the correct notion to consider a biologically meaningful family of update
schedules called {\em block-sequential} in the literature. Since another
work of Aracena {\em et al.}~\cite{afmn11} a precise characterization of the valid labelings
is known, but their combinatorics remained puzzling.
After formal definitions and known results in Sections~\ref{s:def} and~\ref{s:known},
we propose in Section~\ref{s:SPolyc} an explanation for this difficulty,
through the lens of computational complexity theory:
we prove that counting the number of update digraphs
(valid $\{\labplus,\labminus\}$-labelings) is $\SPoly$-complete.
In Section~\ref{s:cacti} we consider the problem restricted to 
the family of oriented cactus graphs and give a $\O(n^2 \log n \log\log n)$ time algorithm,
and finally in Section~\ref{s:decomp} we present a decomposition method leading to
a $\O(n^2 \log^2 n \log\log n)$ algorithm for oriented series-parallel graphs.

\section{Definitions}
\label{s:def}

Given a finite alphabet $[q]=\{1,\dots,q\}$, an {\em automata network} (AN) of
size $n$ is a function $f: [q]^n \to [q]^n$. We denote $x_i$ the {\em
component} $i \in [n]$ of some {\em configuration} $x \in [q]^n$. ANs are more
conveniently seen as $n$ {\em local functions} $f_i: [q]^n \to [q]$ describing
the update of each component, {\em i.e.} with $f_i(x)=f(x)_i$. The {\em
interaction digraph} captures the effective dependencies among components, and
is defined as the digraph $\ID{f}=([n],A_f)$ with 
$$
  (i,j) \in A_f \quad\iff\quad f_j(x) \neq f_j(y) \text{ for some } x,y \in [q]^n
  \text{ with } x_{i'}=y_{i'} \text{ for all } i' \neq i.
$$

It is well known that the schedule of components update may have a great impact
on the dynamics~\cite{ags13,bgps20,f14,gmmrw19,ns18,pmmor19}. 
A {\em block-sequential update schedule} $B=(B_1,\dots,B_t)$ is an ordered
partition of $[n]$, defining the following dynamics
$$
  f^{(B)}=f^{(B_t)} \circ \dots \circ f^{(B_2)} \circ f^{(B_1)}
  \quad\text{with}\quad
  f^{(B_i)}(x)_j=\begin{cases}
    f_j(x) & \text{if } j \in B_i\\
    x_j & \text{if } j \notin B_i
  \end{cases}
$$
{\em i.e.}, parts are updated sequentially one after the other, and components within a part are
updated in parallel. For the parallel update schedule
$B^\texttt{par}=([n])$, we have $f^{(B^\texttt{par})}=f$. Block-sequential
update schedules are a classical family of update schedules considered in the
literature, because they are perfectly fair: every local function is applied
exactly once during each step.  Equipped with an update schedule, $f^{(B)}$ is
a discrete dynamical system on $[q]^n$. In the following we will shortly say
{\em update schedule} to mean {\em block-sequential update schedule}.

\begin{figure}
  \centerline{\includegraphics{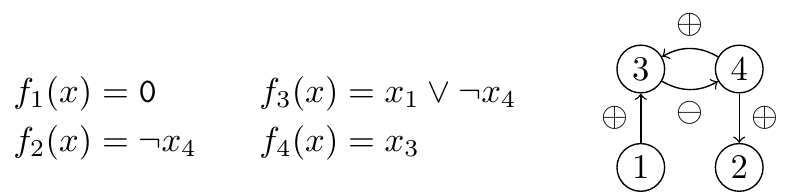}}
  \caption{
    Example of an AN $f$ on the Boolean alphabet $[q]=\{1,2\}$ (conventionally
    renamed $\{\0,\1\}$), its interaction digraph, and a
    $\{\labplus,\labminus\}$-labeling $\lab_B=\lab_{B'}$ corresponding to the
    two equivalent update schedules $B=(\{1,2,3\},\{4\})$ and
    $B'=(\{1,3\},\{2,4\})$.
  }
  \label{fig:ID-UD}
\end{figure}

It turns out quite intuitively that some update schedules will lead to the same
dynamics, when the ordered partitions are very close and the difference relies
on components far apart in the interaction digraph (see an example on
Figure~\ref{fig:ID-UD}). Aracena {\em et al.} introduced in~\cite{agms09}
the notion of {\em update digraph} to capture this fact. To an update schedule
one can associate its {\em update digraph}, which is a
$\{\labplus,\labminus\}$-labeling of the arcs of the interaction digraph of the
AN, such that $(i,j)$ is negative ($\labminus$) when $i$ is updated strictly
before $j$, and positive ($\labplus$) otherwise.
Formally, given an update schedule $B=(B_1,\dots,B_t)$,
$$
  \forall (i,j) \in A_f : \lab_B((i,j))=\begin{cases}
    \labplus & \text{if } i \in B_{t_i} \text{ and } j \in B_{t_j}
    \text{ with } t_i \geq t_j,\\
    \labminus & \text{if } i \in B_{t_i} \text{ and } j \in B_{t_j}
    \text{ with } t_i < t_j.
  \end{cases}
$$

\begin{remark}
  \label{remark:loop}
  Loops are always labeled $\labplus$,
  hence we consider all interaction digraphs {\em loopless}.
\end{remark}
  
The following result has been established: given two update schedules, if the
relative order of updates among all adjacent components are identical, then the
dynamics are identical.

\begin{theorem}[\cite{agms09}]\label{theorem:lab}
  Given an AN $f$ and two update schedules $B,B'$,
  if $\lab_B=\lab_{B'}$ then $f^{(B)}=f^{(B')}$.
\end{theorem}

This leads naturally to an {\em equivalence relation} on update schedules, at
the heart of the present work.

\begin{definition}
  $B \equiv B'$ if and only if $\lab_B=\lab_{B'}$.
\end{definition}

It is very important to note that, though every update schedule corresponds to a
$\{\labplus,\labminus\}$-labeling of $\ID{f}$,
the reciprocal of this fact is not true.
For example, a cycle with all arcs labeled $\labminus$ would lead to a
contradiction where components are updated strictly before themselves.
Aracena {\em et al.} gave a precise
characterization of {\em valid} update digraphs ({\em i.e.} the ones
corresponding to at least one update schedule).

\begin{theorem}[\cite{afmn11}]\label{theorem:lab_valid}
  A labeling function $\lab: A \to \{ \labplus,\labminus \}$ is valid if
  and only if there is no cycle $(i_0,i_1,\dots,i_k)$, with $i_0=i_k$, of
  length $k>0$ such that both:
  \begin{itemize}
    \item $\forall~ 0 \leq j < k: ((i_j,i_{j+1}) \in A \wedge
      \lab((i_j,i_{j+1}))=\labplus) \vee ((i_{j+1},i_j) \in A \wedge
      \lab((i_{j+1},i_j))=\labminus)$,
    \item $\exists~ 0 \leq j < k: \lab((i_{j+1},i_j))=\labminus$.
  \end{itemize}
  In words, the multidigraph where the labeling is unchanged but the
  orientation of negative arcs is reversed, does not contain a cycle with at
  least one negative arc {\em (forbidden cycle)}.
\end{theorem}

As a corollary, one can decide in polynomial time whether a labeling is
valid.\\[.5em]
\decisionpb{Valid Update Digraph Problem (Valid-UD Problem)}
{a labeling $\lab:A\to\{\labplus,\labminus\}$ of a digraph $G=(V,A)$.}
{is $\lab$ valid?}

\begin{corollary}[\cite{afmn11}]
  \label{corollary:valid}
  {\bf Valid-UD Problem} is in $\Poly$.
\end{corollary}

We are interested in the following counting problem.\\[.5em]
\countingpb{Update Digraphs Counting (\#UD)}
{a digraph $G=(V,A)$.}
{$\SUD{G}=|\{ \lab:A \to \{\labminus,\labplus\} \mid \lab \text{ is valid} \}|$.}\\[.5em]

The following definition is motivated by Theorem~\ref{theorem:lab_valid}.

\begin{definition}
  Given a directed graph $G=(V,A)$, let $\un{G}=(V,\un{A})$ denote the undirected
  multigraph underlying $G$,
  {\em i.e.} with an edge $\{i,j\} \in \un{A}$ for each $(i,j) \in A$.
\end{definition}

\begin{remark}
  We can restrict our study to connected digraphs (that is,
  such that $\un{G}$ is connected),
  because according to Theorem~\ref{theorem:lab_valid} the only invalid
  labelings contain (forbidden) cycles. Given some $G$ with $V_1,\dots,V_k$ its
  connected components, and $G[V_i]$ the subdigraph induced by $V_i$, we
  straightforwardly have
  $$
    \SUD{G} = \prod_{i \in [k]} \SUD{G[V_i]},
  $$
  and this decomposition can be computed
  in linear time from folklore algorithms.
\end{remark}

\begin{theorem}
  {\bf \#UD} is in $\SPoly$.
\end{theorem}

\begin{proof}
  The following non-deterministic algorithm runs in polynomial time
  (step~\ref{step:valid} from Corollary~\ref{corollary:valid}), and its number
  of accepting branches equals $\SUD{G}$:
  \begin{enumerate}
    \item guess a labeling $\lab:A\to\{\labplus,\labminus\}$ (polynomial
      space),
    \item\label{step:valid} accept if $\lab$ is valid, otherwise reject.\\[-2.5em]
  \end{enumerate}
\end{proof}

\section{Further known results}
\label{s:known}

The consideration of update digraphs has been initiated by Aracena {\em et al.}
in 2009~\cite{agms09}, with their characterization
(Theorem~\ref{theorem:lab_valid}) in~\cite{afmn11}.
In Section~\ref{s:SPolyc} we will present a problem closely related to {\bf \#UD}
that has been proven to be $\NP$-complete in~\cite{afmn11}, {\bf UD Problem},
and bounds that we can deduce on {\bf \#UD} (Corollary~\ref{coro:bounds},
from~\cite{adfm13b}).
In~\cite{adfm13a} the authors present an algorithm to enumerate update digraphs,
and prove its correctness. They also consider a surprisingly complex question:
given an AN $f$, knowing whether there exist two block-sequential update schedules
$B,B'$ such that $f^{(B)} \neq f^{(B')}$, is $\NP$-complete.
The value of $\SUD{G}$ is known to be
$3^n-2^{n+1}+2$ for bidirected cycles on $n$ vertices~\cite{pmmor19},
and to equal $n!$ if and only if the digraph is a tournament on $n$ vertices~\cite{adfm13b}.

\section{Counting update digraphs is $\SPoly$-complete}
\label{s:SPolyc}

The authors of~\cite{afmn11} have exhibited an insightful relation
between valid labelings and feedback arc sets of a digraph. We recall that a
{\em feedback arc set} (FAS) of $G=(V,A)$ is a subset of arcs $F \subseteq
A$ such that the digraph $(V,A \setminus F)$ is acyclic, and its
size is $|F|$. This relation is developed inside the proof of $\NP$-completeness
of the following decision problem. We reproduce it as a Lemma, with its
argumentation for the sake of comprehension.\\[.5em]
\decisionpb{Update Digraph Problem (UD Problem)}
{a digraph $G=(V,A)$ and an integer $k$.}
{does there exist a valid labeling of size at most $k$?}\\[.5em]
The {\em size} of a labeling is its number of $\labplus$ labels. It is clear
that minimizing the number of $\labplus$ labels (or equivalently maximizing the
number of $\labminus$ labels) is the difficult direction, the contrary being easy
because $\lab(a)=\labplus$ for all $a \in A$ is always valid (and corresponds
to the parallel update schedule $B^\texttt{par}$).

\begin{lemma}[appears in {\cite[Theorem~16]{afmn11}}]
  \label{lemma:labFAS}
  There exists a bijection between minimal valid labelings and minimal feedback
  arc sets of a digraph $G=(V,A)$.
\end{lemma}

\begin{proof}
  To get the bijection, we simply identify a labeling $\lab$ with its set of
  arcs labeled $\labplus$, denoted $F_{\lab}=\{ a \in A \mid
  \lab(a)=\labplus \}$.

  Given any valid labeling $\lab$, $F_{\lab}$ is a FAS, otherwise there is a
  cycle with all its arcs label $\labminus$, which is forbidden.

  Given a minimal FAS $F$, let us now argue that $\lab$ such that $F_{\lab}=F$ is
  a valid labeling. First observe that for every $a \in F$ there is a cycle in $G$
  containing $a$ and no other arc of $F$, otherwise $F$ would not be minimal
  (removing an arc $a$ not fulfilling this observation would give a smaller FAS).
  By contradiction, if $\lab$ creates a forbidden cycle $C$, then to every
  arc $a$ labeled $\labplus$ in $C$ we can associate, from the previous
  observation, a cycle $C_a$ containing no other arc labeled $\labplus$ and
  distinct from $C$. It follows that replacing all such arcs $a$ in $C$ with
  $C_a \setminus \{a\}$ gives a cycle which now contains only arcs labeled
  $\labminus$, {\em i.e.} a cycle (recall that the orientation of $\labminus$
  arcs is reversed in forbidden cycles) with no arc in $F$, a contradiction.
\end{proof}

Any valid labeling corresponds to a FAS, and every minimal FAS corresponds to
a valid labeling, hence the following bounds hold. The strict inequality for
the lower bound comes from the fact that labeling all arcs $\labplus$ does not
give a minimal FAS, as noted in~\cite{adfm13b} where the authors also
consider the relation between update digraphs (valid labelings) and feedback
arc sets, but from another perspective.

\begin{corollary}[\cite{afmn11}]
  \label{coro:bounds}
  For any digraph $G$, let $\#\texttt{FAS}(G)$ and $\#\texttt{MFAS}(G)$ be respectively the number
  of FAS and minimal FAS of $G$, then $\#\texttt{MFAS}(G) < \SUD{G} \leq \#\texttt{FAS}(G)$.
\end{corollary}

From Lemma~\ref{lemma:labFAS} and results on the complexity of FAS counting
problems presented in~\cite{p19} (one of them coming from~\cite{ss02}),
we have the following corollary (minimum FAS are minimal,
hence the identity is a parsimonious reduction from the same
problems on FAS).

\begin{corollary}
  Counting the number of valid labelings of minimal size is $\SPoly$-complete,
  and counting the number of valid labelings of minimum size is
  $\mathsf{\#\!\cdot\!OptP[\log n]}$-complete.
\end{corollary}

However the correspondence given in Lemma~\ref{lemma:labFAS} does not hold in
general: there may exist some FAS $F$ such that $\lab$ with $F_{\lab}=F$ is not
a valid labeling (see Figure~\ref{fig:labFAS} for an example). As a
consequence we do not directly get a counting reduction to {\bf \#UD}.  It nevertheless
holds that {\bf \#UD} is $\SPoly$-hard, with the following reduction.

\begin{figure}
  \centerline{\includegraphics{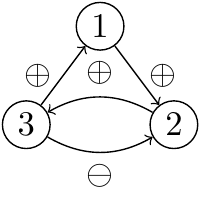}}
  \caption{
    $F=\{(1,2),(2,3),(3,1)\}$ is a FAS, but the corresponding labeling
    is not valid: component $3$ is updated prior to $2$,
    $1$ not prior to $2$, and $3$ not prior to $1$, which is impossible.
  }
  \label{fig:labFAS}
\end{figure}

\begin{theorem}
  \label{theorem:SUD}
  {\bf \#UD} is $\SPoly$-hard.
\end{theorem}

\begin{proof}
  We present a (polynomial time) parsimonious reduction from the problem of
  counting the number of acyclic orientations of an undirected graph, proven to
  be $\SPoly$-hard in~\cite{l86}.

  Given an undirected graph $G=(V,E)$, let $\prec$ denote an arbitrary total order
  on $V$. Construct the digraph $G'=(V,A)$ with
  $A$ the orientation of $E$ according to $\prec$,
  $$
    (u,v) \in A \quad\iff\quad \{u,v\} \in E \text{ and } u \prec v.
  $$
  An example is given on Figure~\ref{fig:A}.
  A key property is that $G'$ is acyclic, because $A$ is constructed from
  an order $\prec$ on $V$ (a cycle would have at least one arc $(u,v)$ with $v
  \prec u$).

  \begin{figure}
    \centerline{\includegraphics{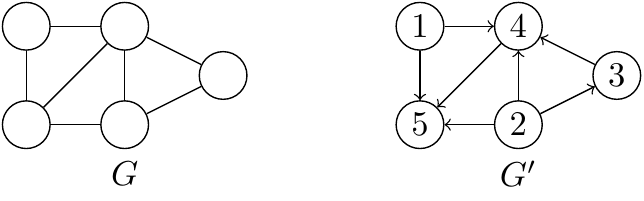}}
    \caption{
      An undirected graph $G$ (instance of acyclic orientation counting),
      and the obtained digraph $G'$ (instance of update digraph counting).
    }
    \label{fig:A}
  \end{figure}

  We claim that there is a bijection between the valid labelings of $G'$ and
  the acyclic orientations of $G$: to a valid labeling
  $\lab:A\to\{\labplus,\labminus\}$ of $G'$ we associate the orientation
  $$
    \begin{array}{rr}
      O= & \{ (u,v) \mid (u,v) \in A \text{ and } \lab((u,v))=\labplus \}\phantom{.}\\
      & \cup~ \{ (v,u) \mid (u,v) \in A \text{ and } \lab((u,v))=\labminus \}.
    \end{array}
  $$

  \begin{figure}
    \centerline{\includegraphics{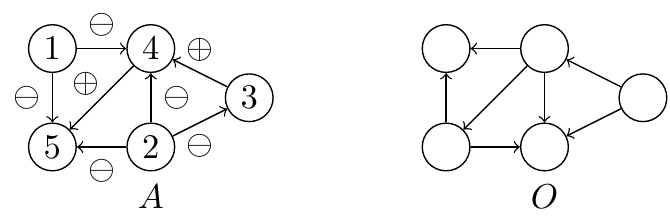}}
    \caption{
      A valid labeling $A$ of $G'$, and the corresponding
      orientation $O$ of $G$.
    }
    \label{fig:O}
  \end{figure}


  First remark that $O$ is indeed an orientation of $E$: each edge of $E$ is
  transformed into an arc of $A$, and each arc of $A$ is
  transformed into an arc of $O$. An example is given on Figure~\ref{fig:O}.
  Now observe that $O$ is exactly obtained from $G'$ by reversing the
  orientation of arcs labeled $\labminus$ by $\lab$. Furthermore, a cycle in
  $O$ must contain at least one arc labeled $\labminus$ by $\lab$, because $G'$
  is acyclic and $\labplus$ labels copy the orientation of $G'$.  The claim
  therefore follows directly from the characterization of
  Theorem~\ref{theorem:lab_valid}.
\end{proof}

\section{Quasi-quadratic time algorithm for oriented cacti}
\label{s:cacti}

The difficulty of counting the number of update digraphs comes from the
interplay between various possible cycles, as is assessed by the parsimonious reduction from
acyclic orientations counting problem to ${\bf \#UD}$.
Answering the problem for an oriented tree with $m$ arcs is for example very simple:
all of the $2^m$ labelings are valid.
Cactus undirected graphs are defined in terms of very restricted entanglement of cycles,
which we can exploit to compute the number of update digraphs for any
orientation of its edges.

\begin{definition}
  A {\em cactus} is a connected undirected graph such that any vertex (or
  equivalently any edge) belongs to at most one simple cycle (cycle without
  vertex repetition). An {\em oriented cactus} $G$ is a digraph such that $\un{G}$
  is a cactus.
\end{definition}

Cacti may intuitively be thought as trees with 
cycles. This is indeed the idea behind the {\em skeleton} of a cactus
introduced in~\cite{bk98}, via the following notions:
\begin{itemize}
  \item a {\em c-vertex} is a vertex of degree two included in exactly one cycle,
  \item a {\em g-vertex} is a vertex not included in any cycle,
  \item remaining vertices are {\em h-vertices},
\end{itemize}
and a {\em graft} is a maximal subtree of g- and h-vertices with no two h-vertices
belonging to the same cycle. Then a cactus can be decomposed as
grafts and cycles (two classes called {\em blocks}),
connected at h-vertices according to a tree skeleton.
These notions directly apply to oriented cacti
(see an example on Figure~\ref{fig:cactus}).

\begin{figure}
  \centerline{\includegraphics{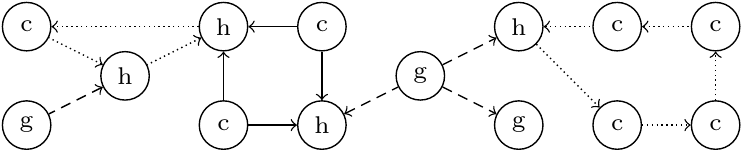}}
  \caption{
    An oriented cactus $G$, with $\{$c,g,h$\}$-vertex labels.
    Graft arcs are dashed, cycles forming directed cycles are dotted,
    and cycles not forming directed cycles are solid.
    Theorem~\ref{theorem:cactus} counts
    $\SUD{G}=2^12^3(2^3-1)(2^5-1)(2^4-2)=48\,608$.
  }
  \label{fig:cactus}
\end{figure}

\begin{theorem}
  \label{theorem:cactus}
  {\bf \#UD} is computable in time $\O(n^2 \log n \log\log n)$ for oriented cacti.
\end{theorem}

\begin{proof}
  The result is obtained from the skeleton of an oriented cactus $G$,
  since potential forbidden cycles are limited to within blocks of the skeleton.
  From this independence, any union of valid labelings on blocks is valid,
  and we have the product
  $$
    \SUD{G}=
    \prod_{H \in \mathcal G} 2^{|H|}
    \prod_{H \in \vec{\mathcal C}} (2^{|H|}-1)
    \prod_{H \in \mathcal C} (2^{|H|}-2)
  $$
  where $\mathcal G$ is the set of grafts of $G$,
  $\vec{\mathcal C}$ is the set of cycles forming directed cycles,
  $\mathcal C$ is the set of cycles not forming directed cycles,
  and $|H|$ is the number of arcs in block $H$.
  Indeed, grafts cannot create forbidden cycles hence any
  $\{\labplus,\labminus\}$-labeling will be valid,
  cycles forming a directed cycle can create exactly one forbidden cycle
  (with $\labminus$ labels on all arcs),
  and cycles not forming a directed cycle can create exactly two forbidden cycles
  (one for each possible direction of the cycle).
  In a first step the skeleton of a cactus can be computed in linear
  time~\cite{bk98}. Then, since the size $n$ of the input is
  equal (up to a constant) to the number of arcs, the size of the output contains
  $\O(n)$ bits (upper bounded by the number of
  $\{\labplus,\labminus\}$-labelings), thus naively we have $\O(n)$ terms,
  each of $\O(n)$ bits, and the $\O(n \log n \log\log n)$
  Sch\"onhage–Strassen integer multiplication algorithm~\cite{ss71}
  gives the result.
\end{proof}

\begin{remark}
  \label{remark:mult}
  Assuming multiplications to be done in constant time in the above result
  would be misleading, because we are multiplying integers having a number of
  digits in the magnitude of the input size. Also, the result may be slightly
  strengthened by considering the $\O(n \log n~2^{2 \log^* n})$
  algorithm by F\"urer in 2007~\cite{f07}, and maybe
  more efficient integer multiplication algorithms in the future, such as the
  $\O(n \log n)$ algorithm recently claimed~\cite{hvdh19}.
\end{remark}

\section{Series-parallel decomposition method}
\label{s:decomp}

In this section we present a divide and conquer method in order to solve {\bf
\#UD}, {\em i.e.} in order to count the number of valid labelings (update
digraphs) of a given digraph.
What will be essential in this decomposition method is not the orientation of arcs,
but rather the topology of the underlying undirected (multi)graph $\un{G}$.
The (de)composition is based on defining two endpoints on our digraphs,
and composing them at their endpoints.
It turns out to be closely related to series-parallel graphs
first formalized to model electric networks in 1892~\cite{mm92}.
In Subsection~\ref{ss:comp} we present the operations of composition,
and in Subsection~\ref{ss:sp} we show how it applies to
the family of oriented series-parallel graphs.

\subsection{Sequential, parallel, and free compositions}
\label{ss:comp}

Let us first introduce some notations and terminology on the characterization
of valid labelings provided by Theorem~\ref{theorem:lab_valid}.
Given $\lab:A\to\{\labplus,\labminus\}$, we denote
$\revneg{G}_{\lab}=(V,\revneg{A})$ the multidigraph obtained by reversing the
orientation of negative arcs:
$$
  (i,j) \in \revneg{A} \quad\iff\quad \begin{array}{r}
    (i,j) \in A \text{ and } \lab((i,j))=\labplus,\\
    \text{or } (j,i) \in A \text{ and } \lab((j,i))=\labminus.
  \end{array}
$$
For simplicity we abuse the notation and still denote $\lab$ the labeling of
the arcs of $\revneg{G}_{\lab}$ (arcs keep their label from $G$ to
$\revneg{G}_{\lab}$). From Theorem~\ref{theorem:lab_valid}, $\lab$ is a valid
labeling if and only if $\revneg{G}_{\lab}$ does not contain any cycle with at
least one arc labeled $\labminus$, called {\em forbidden cycle}
(it may contain cycles with all arcs labeled $\labplus$).
A path from $i$ to $j$ in $\revneg{G}_{\lab}$ is called {\em
negative} if it contains at least one arc labeled $\labminus$, and {\em
positive} otherwise.

\begin{definition}
  A {\em source-sink labeled graph (ss-graph)} $(G,\alpha,\beta)$
  is a multigraph $G$ with two distinguished vertices $\alpha \neq \beta$.
  A triple $(G,\alpha,\beta)$ with $G$ a digraph
  such that $(\un{G},\alpha,\beta)$ is a ss-graph,
  is called an {\em oriented ss-graph (oss-graph)}.
\end{definition}

We can decompose the set of update digraphs
(denoted $\UD{G}=\{ \lab:A\to\{\labplus,\labminus\} \mid \lab
\text{ is valid}\}$) into an oss-graph $(G,\alpha,\beta)$,
based on the follow sets.
$$\begin{array}{rl}
  \UDto{G}{\alpha}{\beta}{+} =\{ \lab \in \UD{G} \mid
    &\text{there {\bf exists} a path from } \alpha \text{ to } \beta \text{ in } \revneg{G}_{\lab},\\
    &\text{and {\bf all} paths from } \alpha \text{ to } \beta \text{ in } \revneg{G}_{\lab} \text{ are {\bf positive}}\}\\[.5em]
  \UDto{G}{\alpha}{\beta}{-} =\{ \lab \in \UD{G} \mid
    &\text{there {\bf exists} a {\bf negative} path from } \alpha \text{ to } \beta \text{ in } \revneg{G}_{\lab}\}\\[.5em]
  \UDto{G}{\alpha}{\beta}{\varnothing} =\{ \lab \in \UD{G} \mid
    &\text{there exist {\bf no path} from } \alpha \text{ to } \beta \text{ in } \revneg{G}_{\lab}\}
\end{array}$$
We define analogously $\UDto{G}{\beta}{\alpha}{+}$, $\UDto{G}{\beta}{\alpha}{-}$,
$\UDto{G}{\beta}{\alpha}{\varnothing}$, and partition $\UD{G}$ as:
\begin{multicols}{2}
  \begin{enumerate}
    \item $\UDp{G}{\alpha,\beta}{+,+}=
      \UDto{G}{\alpha}{\beta}{+} \cap \UDto{G}{\beta}{\alpha}{+}$
    \item $\UDp{G}{\alpha,\beta}{+,\varnothing}=
      \UDto{G}{\alpha}{\beta}{+} \cap \UDto{G}{\beta}{\alpha}{\varnothing}$
    \item $\UDp{G}{\alpha,\beta}{-,\varnothing}=
      \UDto{G}{\alpha}{\beta}{-} \cap \UDto{G}{\beta}{\alpha}{\varnothing}$
    \item $\UDp{G}{\alpha,\beta}{\varnothing,+}=
      \UDto{G}{\alpha}{\beta}{\varnothing} \cap \UDto{G}{\beta}{\alpha}{+}$
    \item $\UDp{G}{\alpha,\beta}{\varnothing,-}=
      \UDto{G}{\alpha}{\beta}{\varnothing} \cap \UDto{G}{\beta}{\alpha}{-}$
    \item $\UDp{G}{\alpha,\beta}{\varnothing,\varnothing}=
      \UDto{G}{\alpha}{\beta}{\varnothing} \cap \UDto{G}{\beta}{\alpha}{\varnothing}$
  \end{enumerate}
\end{multicols}
\noindent Notice that the three missing combinations, $\UDp{G}{\alpha,\beta}{+,-}$,
$\UDp{G}{\alpha,\beta}{-,+}$ and $\UDp{G}{\alpha,\beta}{-,-}$, would always be
empty because such labelings contain a forbidden cycle.
For convenience let us denote
$S=\{(+,+),(+,\varnothing),(-,\varnothing),(\varnothing,+),(\varnothing,-),(\varnothing,\varnothing)\}$.
Given any oss-graph
$(G,\alpha,\beta)$ we have
\begin{align}
  \label{eq:partition}
  \SUD{G}=\sum_{(s,t) \in S} \SUDp{G}{\alpha,\beta}{s,t}
\end{align}
where $\SUDp{G}{\alpha,\beta}{s,t}=|\UDp{G}{\alpha,\beta}{s,t}|$. Oss-graphs may
be thought as black boxes, we will compose them using the values of
$\SUDp{G}{\alpha,\beta}{s,t}$, regardless of their inner topologies.

\begin{definition}
  \label{def:comp}
  We define three types of compositions (see Figure~\ref{fig:comp}).
  \begin{itemize}
    \item The {\em series composition} of two oss-graphs
      $(G,\alpha,\beta)$ and $(G',\alpha',\beta')$
      with $V \cap V' = \emptyset$, is the oss-graph
      $\seqcomp{(G,\alpha,\beta)}{(G',\alpha',\beta')}=(D,\alpha,\beta')$ with
      $D$ the one-point join of $G$ and $G'$ identifying components
      $\beta,\alpha'$ as one single component.
    \item The {\em parallel composition} of two oss-graphs
      $(G,\alpha,\beta)$ and $(G',\alpha',\beta')$
      with $V \cap V' = \emptyset$, is the oss-graph
      $\parcomp{(G,\alpha,\beta)}{(G',\alpha',\beta')}=(D,\alpha,\beta)$ with
      $D$ the two-points join of $G$ and $G'$ identifying components
      $\alpha,\alpha'$ and $\beta,\beta'$ as two single components.
    \item The {\em free composition at $v,v'$} of an oss-graph
      $(G=(V,A),\alpha,\beta)$ and a digraph $G'=(V',A')$ with
      $V \cap V'=\emptyset$, $v \in V$, $v' \in V'$, is the oss-graph
      $\freecomp{(G,\alpha,\beta)}{G'}=(D,\alpha,\beta)$ with
      $D$ the one-point join of $G$ and $G'$ identifying 
      $v,v'$ as one single component.
  \end{itemize}
\end{definition}
\noindent Remark that the three types of compositions from Definition~\ref{def:comp}
also apply to (undirected) ss-graph $(\un{G},\alpha,\beta)$.
Series and free compositions differ on the endpoints of the
obtained oss-graph, which has important consequences on counting the number of
update digraphs, as stated in the following results.
We will see in Theorem~\ref{theorem:spf} from Subsection~\ref{ss:sp}
that both series and free compositions are needed in order to decompose
the family of (general) {\em oriented series-parallel graphs}
(to be defined).

\begin{figure}
  \centerline{
    \includegraphics[scale=.7]{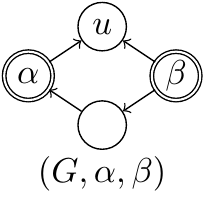}~
    \includegraphics[scale=.7]{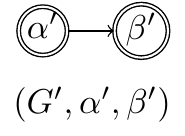}\qquad
    \includegraphics[scale=.7]{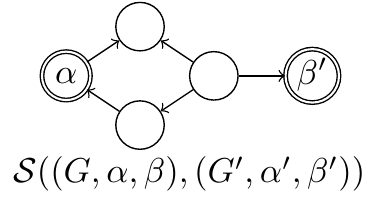}\quad
    \includegraphics[scale=.7]{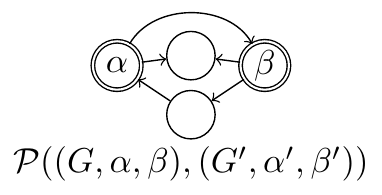}\quad
    \includegraphics[scale=.7]{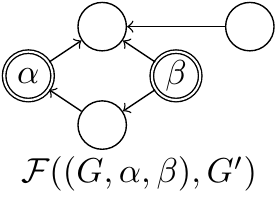}
  }
  \caption{Example of series and parallel compositions, and a free composition at $u,\beta'$.}
  \label{fig:comp}
\end{figure}

\begin{lemma}
  \label{lemma:seqcomp}
  For $(D,\alpha,\beta')=\seqcomp{(G,\alpha,\beta)}{(G',\alpha',\beta')}$,
  the values of $\SUDp{D}{\alpha,\beta'}{s,t}$ for all $(s,t) \in S$ can be
  computed in time $\O(n \log n \log\log n)$ (with $n$ the binary length of the values)
  from the values of $\SUDp{G}{\alpha,\beta}{s,t}$ and $\SUDp{G'}{\alpha',\beta'}{s,t}$
  for all $(s,t) \in S$.
\end{lemma}

\begin{proof}
  The result is obtained by considering the $36$ couples
  of some $\UDp{G}{\alpha,\beta}{s,t}$ with $(s,t) \in S$
  and some $\UDp{G'}{\alpha',\beta'}{s',t'}$ with $(s',t')$.
  For example, when we take the union of any labeling $\lab$ in
  $\UDp{G}{\alpha,\beta}{+,\varnothing}$
  (in $\revneg{G}_{\lab}$ there is at least one path from $\alpha$ to $\beta$,
  all paths from $\alpha$ to $\beta$ positive,
  and no path from $\beta$ to $\alpha$)
  and any labeling $\lab'$ in
  $\UDp{G'}{\alpha',\beta'}{-,\varnothing}$
  (in $\revneg{G'}_{\lab'}$ there is a negative path from $\alpha'$ to $\beta'$,
  and no path from $\beta'$ to $\alpha'$),
  we obtain a valid labeling of $D$ with at least one negative path from $\alpha$ to $\beta'$,
  and no path from $\beta'$ to $\alpha$, hence an element of
  $\UDp{D}{\alpha,\beta'}{-,\varnothing}$. We can deduce that
  $\UDp{D}{\alpha,\beta'}{-,\varnothing}$ contains these
  $\SUDp{G}{\alpha,\beta}{+,\varnothing}\SUDp{G'}{\alpha',\beta'}{-,\varnothing}$
  labelings.
  The results of this simple reasoning in all $36$ cases is presented
  on Table~\ref{table:seqcomp}.
  \begin{table}
    \centering
    \begin{tabular}{ | c | c | c | c | c | c | c | c | c | c | }
      \hline
      $\mathcal S$ & \cellcolor[gray]{\tr} $+,+$ & \cellcolor[gray]{\tr} $+,\varnothing$ & \cellcolor[gray]{\tr} $-,\varnothing$ & \cellcolor[gray]{\tr} $\varnothing,+$ & \cellcolor[gray]{\tr} $\varnothing,-$ & \cellcolor[gray]{\tr} $\varnothing,\varnothing$ \\
      \hline 
      \cellcolor[gray]{\tr} $+,+$ & $+,+$ &  $+,\varnothing$ & $-,\varnothing$ &  $\varnothing,+$ & $\varnothing,-$ & $\varnothing,\varnothing$ \\
      \hline 
      \cellcolor[gray]{\tr} $+,\varnothing$ & $+,\varnothing$ &  $+,\varnothing$  & $-,\varnothing$ &  $\varnothing,\varnothing$ & $\varnothing,\varnothing$ & $\varnothing,\varnothing$ \\
      \hline 
      \cellcolor[gray]{\tr} $-,\varnothing$ & $-,\varnothing$  &  $-,\varnothing$ & $-,\varnothing$ &  $\varnothing,\varnothing$ & $\varnothing,\varnothing$ & $\varnothing,\varnothing$ \\
      \hline 
      \cellcolor[gray]{\tr} $\varnothing,+$ & $\varnothing,+$  &  $\varnothing,\varnothing$ & $\varnothing,\varnothing$ &  $\varnothing,+$ & $\varnothing,-$ & $\varnothing,\varnothing$ \\
      \hline 
      \cellcolor[gray]{\tr} $\varnothing,-$ & $\varnothing,-$  &  $\varnothing,\varnothing$  & $\varnothing,\varnothing$ &  $\varnothing,-$ & $\varnothing,-$ & $\varnothing,\varnothing$ \\
      \hline 
      \cellcolor[gray]{\tr} $\varnothing,\varnothing$ & $\varnothing,\varnothing$  &  $\varnothing,\varnothing$ & $\varnothing,\varnothing$ &  $\varnothing,\varnothing$ & $\varnothing,\varnothing$ & $\varnothing,\varnothing$ \\
      \hline
    \end{tabular}
    \caption{
      Row $(s,t)$ corresponds to $\UDp{G}{\alpha,\beta}{s,t}$,
      column $(s',t')$ corresponds to $\UDp{G'}{\alpha',\beta'}{s',t'}$,
      and cell content $(s'',t'')$ indicates that
      the union of any labeling in $\UDp{G}{\alpha,\beta}{s,t}$
      and any labeling in $\UDp{G'}{\alpha',\beta'}{s',t'}$
      gives a labeling in $\UDp{D}{\alpha,\beta'}{s'',t''}$ for
      $(D,\alpha,\beta')=\seqcomp{(G,\alpha,\beta)}{(G',\alpha',\beta')}$.
    }
    \label{table:seqcomp}
  \end{table}
  This allows to compute each value $\SUDp{D}{\alpha,\beta'}{s'',t''}$
  as the finite sum, for each line $(s,t)$ and column $(s',t')$
  where $(s'',t'')$ appears, of the term
  $\SUDp{G}{\alpha,\beta}{s,t}\SUDp{G'}{\alpha',\beta'}{s',t'}$.
  As an example we have $\SUDp{D}{\alpha,\beta'}{+,\varnothing}=
  \SUDp{G}{\alpha,\beta}{+,+}\SUDp{G'}{\alpha',\beta'}{+,\varnothing}+
  \SUDp{G}{\alpha,\beta}{+,\varnothing}\SUDp{G'}{\alpha',\beta'}{+,+}+
  \SUDp{G}{\alpha,\beta}{+,\varnothing}\SUDp{G'}{\alpha',\beta'}{+,\varnothing}$.
  Summations on $n$ bits are performed in linear time with schoolbook
  algorithm, and multiplications on $n$ bits are performed in time
  $\O(n \log n \log\log n)$ from Sch\"onhage–Strassen algorithm~\cite{ss71}.
\end{proof}

\begin{lemma}
  \label{lemma:parcomp}
  For $(D,\alpha,\beta)=\parcomp{(G,\alpha,\beta)}{(G',\alpha',\beta')}$,
  the values of $\SUDp{D}{\alpha,\beta}{s,t}$ for all $(s,t) \in S$ can be
  computed in time $\O(n \log n \log\log n)$ (with $n$ the binary length of the values)
  from the values of $\SUDp{G}{\alpha,\beta}{s,t}$ and $\SUDp{G'}{\alpha',\beta'}{s,t}$
  for all $(s,t) \in S$.
\end{lemma}

\begin{proof}
  The proof is analogous to Lemma~\ref{lemma:seqcomp}, with Table~\ref{table:parcomp}.
  Remark that in this case, it is possible to create invalid labelings for $D$
  (containing some forbidden cycle) from the union of two valid labelings for
  $G$ and $G'$.
  For example, when we take the union of any labeling $\lab$ in
  $\UDp{G}{\alpha,\beta}{+,+}$
  (in $\revneg{G}_{\lab}$ there is at least one path from $\alpha$ to $\beta$,
  all paths from $\alpha$ to $\beta$ positive,
  and the same from $\beta$ to $\alpha$)
  and any labeling $\lab'$ in
  $\UDp{G'}{\alpha',\beta'}{-,\varnothing}$
  (in $\revneg{G'}_{\lab'}$ there is a negative path from $\alpha'$ to $\beta'$,
  and no path from $\beta'$ to $\alpha'$),
  we obtain an invalid labeling of $D$ because the concatenation of
  a positive path from $\alpha$ to $\beta$ with a negative path from
  $\beta'$ to $\alpha'$ gives a forbidden cycle in $D$
  (recall that $\alpha,\alpha'$ and $\beta,\beta'$ are identified).
  \begin{table}
    \centering
    \begin{tabular}{ | c | c | c | c | c | c | c | c | c | c | }
      \hline
      $\mathcal P$ & \cellcolor[gray]{\tr} $+,+$ & \cellcolor[gray]{\tr} $+,\varnothing$ & \cellcolor[gray]{\tr} $-,\varnothing$ & \cellcolor[gray]{\tr} $\varnothing,+$ & \cellcolor[gray]{\tr} $\varnothing,-$ & \cellcolor[gray]{\tr} $\varnothing,\varnothing$ \\
      \hline 
      \cellcolor[gray]{\tr} $+,+$ & $+,+$ & $+,+$ & & $+,+$ & & $+,+$ \\
      \hline 
      \cellcolor[gray]{\tr} $+,\varnothing$ & $+,+$ & $+,\varnothing$ & $-,\varnothing$ & $+,+$ & & $+,\varnothing$ \\
      \hline 
      \cellcolor[gray]{\tr} $-,\varnothing$ & & $-,\varnothing$ & $-,\varnothing$ & & & $-,\varnothing$ \\
      \hline 
      \cellcolor[gray]{\tr} $\varnothing,+$ & $+,+$ & $+,+$ & & $\varnothing,+$ & $\varnothing,-$ & $\varnothing,+$ \\
      \hline 
      \cellcolor[gray]{\tr} $\varnothing,-$ & & & & $\varnothing,-$ & $\varnothing,-$ & $\varnothing,-$ \\
      \hline 
      \cellcolor[gray]{\tr} $\varnothing,\varnothing$ & $+,+$ & $+,\varnothing$ & $-,\varnothing$ & $\varnothing,+$ & $\varnothing,-$ & $\varnothing,\varnothing$ \\
      \hline
    \end{tabular}
    \caption{
      Row $(s,t)$ corresponds to $\UDp{G}{\alpha,\beta}{s,t}$,
      column $(s',t')$ corresponds to $\UDp{G'}{\alpha',\beta'}{s',t'}$,
      and cell content $(s'',t'')$ indicates that
      the union of any labeling in $\UDp{G}{\alpha,\beta}{s,t}$
      and any labeling in $\UDp{G'}{\alpha',\beta'}{s',t'}$
      gives a labeling in $\UDp{D}{\alpha,\beta}{s'',t''}$ for
      $(D,\alpha,\beta)=\parcomp{(G,\alpha,\beta)}{(G',\alpha',\beta')}$.
      Empty cells mean that unions give invalid labelings.
    }
    \label{table:parcomp}
  \end{table}
\end{proof}

Note that Remark~\ref{remark:mult} also applies to Lemmas~\ref{lemma:seqcomp}
and~\ref{lemma:parcomp}.
For the free composition the count is easier.

\begin{lemma}
  \label{lemma:freecomp}
  For $(D,\alpha,\beta)=\freecomp{(G,\alpha,\beta)}{G'}$,
  we have $\SUDp{D}{\alpha,\beta}{s,t}=\SUDp{G}{\alpha,\beta}{s,t}\SUD{G'}$
  for all $(s,t) \in S$.
\end{lemma}

\begin{proof}
  The endpoints of the oss-graph $(D,\alpha,\beta)$ are the endpoints of the
  oss-graph $(G,\alpha,\beta)$, and it is not possible to create a forbidden cycle
  in the union of a valid labeling on $G$ and a valid labeling of $G'$,
  therefore the union is always a valid labeling of $D$, each one belonging to
  the part $(s,t)$ of $(D,\alpha,\beta)$
  corresponding to the part $(s,t) \in S$ of $(G,\alpha,\beta)$.
\end{proof}

\subsection{Application to oriented series-parallel graphs}
\label{ss:sp}

The series and parallel compositions of Definition~\ref{def:comp}
correspond exactly to the class of two-terminal series-parallel graphs
from~\cite{a61,rs42,vtl79}.

\begin{definition}
  \label{def:ttsp}
  A ss-graph $(G,\alpha,\beta)$ is {\em two-terminal series-parallel (a ttsp-graph)}
  if and only if one the following holds.
  \begin{itemize}
    \item $(G,\alpha,\beta)$ is a {\em base ss-graph}
      with two vertices $\alpha,\beta$ and one edge $\{\alpha,\beta\}$.
    \item $(G,\alpha,\beta)$ is obtained by a series or parallel
      composition\footnote{With Definition~\ref{def:comp} applied to (undirected) ss-graphs.}
      of two ttsp-graphs.
  \end{itemize}
  In this case $G$ alone is called a {\em blind} ttsp-graph.
\end{definition}
%
%

Adding the free composition allows to go from two-terminal series-parallel graphs
to (general) series-parallel graphs~\cite{d65,vtl79}.
More precisely, it allows exactly to add tree structures to ttsp-graphs, as we argue now
(ttsp-graphs do not contain arbitrary trees, its only acyclic graphs being simple paths;
{\em e.g.} one cannot build a {\em claw} from Definition~\ref{def:ttsp}).

\begin{definition}
  \label{def:sp}
  A multigraph $G$ is {\em series-parallel (sp-graph)} if and only if
  all its 2-connected components are blind ttsp-graphs.
  A digraph $G$ such that $\un{G}$ is an sp-graph,
  is called an {\em oriented sp-graph (osp-graph)}.
\end{definition}

The family of sp-graphs corresponds to the multigraphs obtained by series, parallel
and free compositions from base ss-graphs.
See Figure~\ref{fig:sp} for an example $G$
of osp-graph illustrating Theorem~\ref{theorem:spf},
and the decomposition method to compute $\SUD{G}$ (Theorem~\ref{theorem:osp}).

\begin{figure}
  \centerline{
    \raisebox{1cm}{\includegraphics{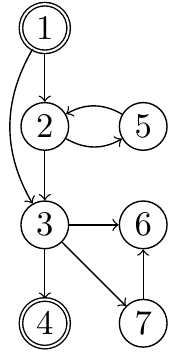}}
    \hspace*{1.5cm}
    \includegraphics{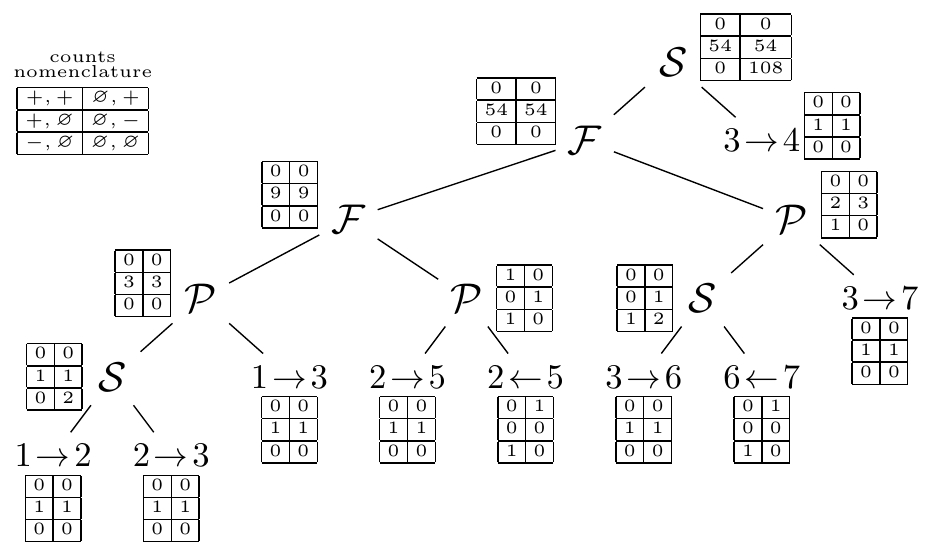}
  }
  \caption{
    An example $G$ of osp-graph (left),
    and an illustration of the decomposition
    method (Theorem~\ref{theorem:osp}) computing $\SUD{G}=216$ based on the values
    of $\SUDp{G}{\alpha,\beta}{s,t}$ for all $(s,t) \in S$,
    from the oriented base ss-graphs (leaf nodes) to $G$ (root of the decomposition tree).
  }
  \label{fig:sp}
\end{figure}

\begin{theorem}
  \label{theorem:spf}
  $G$ is an sp-graph if and only if $(G,\alpha,\beta)$ is
  obtained by series, parallel and free compositions from base ss-graphs, for
  some $\alpha,\beta$.
\end{theorem}

\begin{proof}[Proof sketch (see Appendix~\ref{a:sp})]
  Free compositions allow to build all sp-graphs, because it offers the
  possibility to create the missing tree structures of ttsp-graphs:
  arbitrary 1-connected components
  linking 2-connected ttsp-graphs.
  Moreover free compositions do not go beyond sp-graphs,
  since the obtained multigraphs still have treewidth 2.
\end{proof}

\begin{theorem}
  \label{theorem:osp}
  {\bf \#UD} is solvable in time $\O(n^2 \log^2 n \log\log n)$ on osp-graphs (without promise).
\end{theorem}

\begin{proof}[Proof sketch (see Appendix~\ref{a:sp})]
  This is a direct consequence of Lemmas~\ref{lemma:seqcomp},
  \ref{lemma:parcomp} and~\ref{lemma:freecomp},
  because all values are in $\O(2^n)$
  (the number of $\{\labplus,\labminus\}$-labelings)
  hence on $\O(n)$ bits, the values of $\SUDp{G}{\alpha,\beta}{s,t}$
  are trivial for oriented base ss-graphs,
  and we perform $\O(n \log n)$ compositions (to reach Formula~\ref{eq:partition}).
  The absence of promise comes from a linear time recognition algorithm
  in~\cite{vtl79} for ttsp-graphs, which also provides the decomposition structure.
\end{proof}

Again, Remark~\ref{remark:mult} applies to Theorem~\ref{theorem:osp}. 

\section{Conclusion}

Our main result is the $\SPoly$-completeness of {\bf \#UD}, {\em i.e.} of
counting the number of non-equivalent block-sequential update schedules of a
given AN $f$. We proved that this count can nevertheless be done in $\O(n^2
\log n \log\log n)$ time for oriented cacti, and in $\O(n^2 \log^2 n \log\log n)$
time for oriented series-parallel graphs. This last result has been obtained
via a decomposition method providing a divide-and-conquer algorithm.

Remark that cliques or tournaments are intuitively difficult instances of {\bf \#UD},
because of the intertwined structure of potential forbidden cycles. It turns
out that $K_4$ is the smallest clique that cannot be build with series,
parallel and free decompositions, and that series-parallel graphs
(Definition~\ref{def:sp}) correspond exactly to the family of $K_4$-minor-free
graphs~\cite{d65} (it is indeed closed by
minor~\cite{rs04}).
In further works we would like to extend this caracterization and
the decomposition method to (di)graphs with multiple endpoints.

The complexity analysis of the algorithms presented in Theorems~\ref{theorem:cactus}
and~\ref{theorem:osp} may be improved, and
adapted to the parallel setting using the algorithms presented
in~\cite{bf96,j92}.
One may also ask for which other classes of digraphs is $\SUD{G}$ computable
efficiently (in polynomial time)? Since we found such an algorithm for graphs
of treewidth 2, could it be that the problem is fixed parameter tractable on
bounded treewidth digraphs? Rephrased more directly, could a general tree
decomposition (which, according to the proof of Theorem~\ref{theorem:spf}, is
closely related to the series-parallel decomposition for treewidth 2) be
exploited to compute the solution to {\bf \#UD}? Alternatively, what other
types of decompositions one can consider in order to ease the computation of
$\SUD{G}$?

Finally, from the multiplication obtained for one-point join of two graphs
(Lemma~\ref{lemma:freecomp} on free composition), we may ask whether $\SUD{G}$
is an evaluation of the Tutte polynomial? From its universality~\cite{b72}, it
remains to know whether there is a deletion-contradiction reduction.
However defining a Tutte polynomial for directed graphs is still an active area of
research~\cite{ab20,c18,pvp16,y19}.

\paragraph{Acknowledgments} This work was mainly funded by our salaries as French 
agents (affiliated to
Laboratoire Cogitamus (CN),
Aix-Marseille Univ, Univ. de Toulon, CNRS, LIS, UMR 7020, Marseille, France
(KP, SS and LV),
Univ. C\^{o}te d'Azur, CNRS, I3S, UMR 7271, Sophia Antipolis, France
(KP),
and
{\'E}cole normale sup{\'e}rieure de Lyon, Computer Science department, Lyon, France (LV)),
and secondarily by the 
projects ANR-18-CE40-0002 FANs, ECOS-CONICYT C16E01, STIC AmSud 19-STIC-03 (Campus 
France 43478PD).

\bibliographystyle{plain}
\bibliography{biblio}

\appendix

\section{Full proofs for the application of the decomposition method to oriented series-parallel graphs}
\label{a:sp}

The concept of decomposition tree associated to ttsp-graphs,
and generalized to sp-graphs,
will be useful to formalize the reasonings leading to
Theorems~\ref{theorem:spf} and~\ref{theorem:osp}.

The structure of a ttsp-graph $(G=(V,E),\alpha,\beta)$ obtained from base
ss-graphs by series and parallel compositions is expressed in a {\em ttsp-tree}
$T_G=(N,F,\sigma:N \to V \times V)$. It is a rooted binary tree, in which each
node $\eta \in N$ has one of the types {\em s-node}, {\em p-node}, {\em
leaf-node}, and a label $\sigma(\eta)$ consisting in a pair of vertices from $G$.
Every node $\eta \in N$ corresponds to a unique ttsp-graph $(G',\alpha',\beta')$
which is a subgraph of $G$, with $\sigma(\eta)=(\alpha',\beta')$. The leaves of
the tree are of type leaf-node and correspond to base ss-graphs, they are in
one-to-one correspondence with the edges of $G$.
The ss-graph associated to an s-node is the series composition of its (ordered)
children, and the ss-graph associated to a p-node is the parallel composition of
its (unordered) children.
It is worth noticing that each ttsp-graph corresponds to at least one ttsp-tree
(and possibly to many ttsp-trees, even non isomorphic ones~\cite[Section 2.2]{vtl79}),
and that each ttsp-tree corresponds to a unique ttsp-graph.

To this classical definition (presented for example in~\cite{vtl79,f97}), we add
{\em f-nodes} corresponding to free compositions, leading to {\em sp-trees}, as follows.
The subtlety is that the two distinguished vertices of one children of a free
composition are discarded. Let $(G=(V,E),\alpha,\beta)$ and
$(G'=(V',E'),\alpha',\beta')$ be the (ordered) children of an f-node $\eta$ of label
$\sigma(\eta)=(v,v')$ with $v \in V$ and $v' \in V'$, the ss-graph associated to
this f-node is the free composition
$\freecomp{(G,\alpha,\beta)}{G'}$ (its label does not correspond to the
distinguished vertices of the ss-graph).
Every node $\eta \in N$ still corresponds to a unique ttsp-graph
$(G',\alpha',\beta')$ which is a subgraph of $G$, with
$\sigma(\eta)=(\alpha',\beta')$ for all but f-nodes; and for f-nodes the
distinguished vertices are that of its first (left on our figures) child.
We have the same correspondence between sp-graphs and sp-trees.
An example is given on Figure~\ref{fig:sptree}.

\begin{figure}
  \centerline{
    \includegraphics{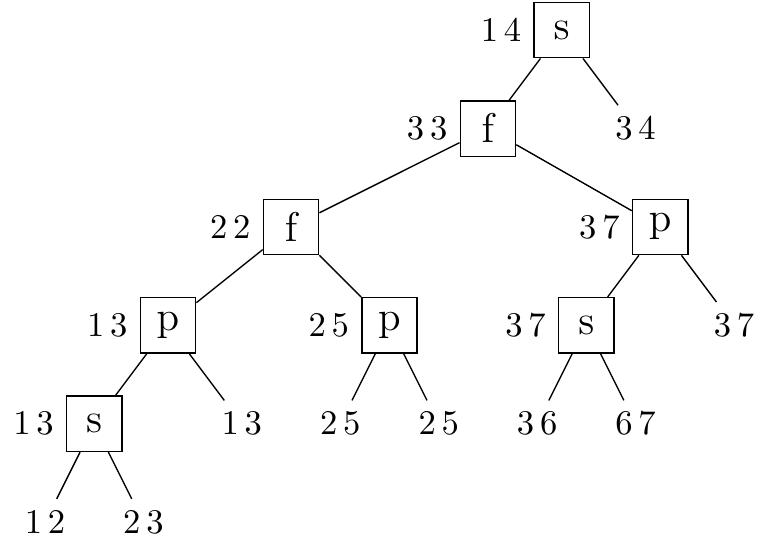}
    \hspace*{1cm}
    \raisebox{1.5cm}{\includegraphics{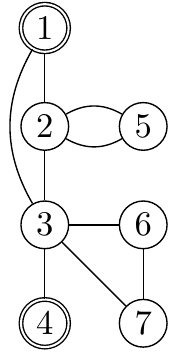}}
  }
  \caption{
    Example sp-tree (left) and the corresponding sp-graph (right).
  }
  \label{fig:sptree}
\end{figure}

Now remark that all free
compositions may be performed at the end of the composition process, {\em i.e.}
on top of the sp-tree. Indeed, free compositions can be inductively pushed
towards the root of an sp-tree using the operations presented on
Figure~\ref{fig:pushfree}, which leave the sp-graph unchanged. When all free
compositions are on top of an sp-tree, we say that it is in {\em
free-on-top normal form}. Thus to any sp-graph corresponds at least one sp-tree
in free-on-top normal form.

\begin{figure}
  \centerline{
    \raisebox{1.5cm}{\includegraphics{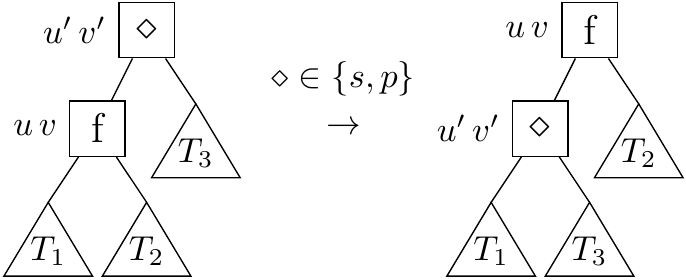}}
    \hspace*{0cm}
    \includegraphics{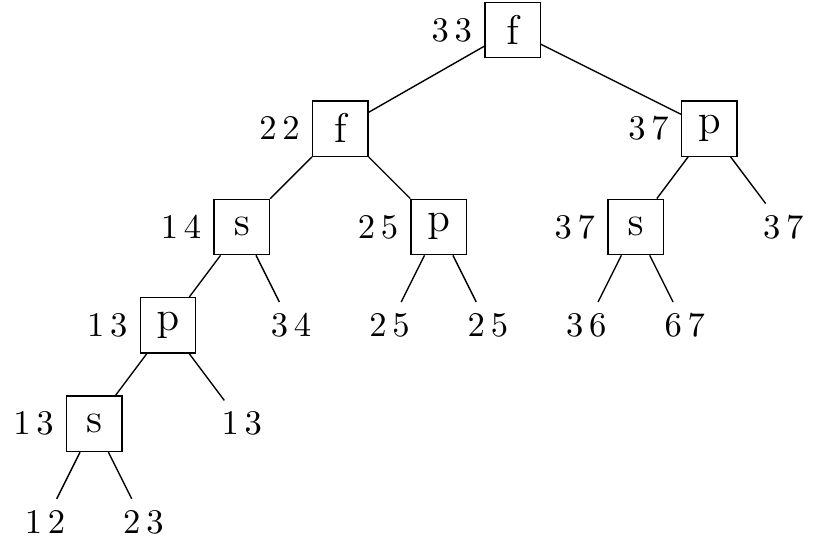}
  }
  \caption{
    Two operations on sp-trees, pushing free compositions towards the root (left).
    The resulting sp-tree corresponds to the same sp-graph, leading by
    induction to an sp-tree in free-on-top normal form.
    Sp-tree in free-on-top normal form (right)
    corresponding to the sp-graph of Figure~\ref{fig:sptree}.
  }
  \label{fig:pushfree}
\end{figure}

\begin{proof}[Proof of Theorem~\ref{theorem:spf}]
  For the ``only if'' part, let $G$ be an sp-graph, with $G_1,\dots,G_x$ its
  2-connected components, $H_1,\dots,H_y$ its remaining connected components
  (which are therefore trees), and $w_1,\dots,w_z$ the set of vertices belonging
  to two 2-connected components among $G_1,\dots,G_x$.
  By definition $G_1,\dots,G_x$ are blind ttsp-graphs, to which we can
  respectively associate some ttsp-trees $T_{G_1},\dots,T_{G_x}$, which are
  particular cases of sp-trees.
  Now, take an edge $\{u,v\}$ of $H_i$ sharing:
  \begin{itemize}
    \item only one vertex with some $T_{G_j}$, and consider the free
      composition of $T_{G_j}$ with the base ss-graph corresponding to edge
      $\{u,v\}$ (joining them at that vertex), we get an sp-tree replacing
      $T_{G_j}$,
    \item its two vertices with respectively some $T_{G_j}$ and $T_{G_{j'}}$,
      and consider first the free composition of $T_{G_j}$ with the base
      ss-graph corresponding to edge $\{u,v\}$ (joining them at $u$), this
      gives $T'_{G_j}$, and second the free composition of $T'_{G_j}$ with
      $T_{G_{j'}}$ (joining them at $v$), we get an sp-tree replacing both
      $T_{G_j}$ and $T_{G_{j'}}$,
  \end{itemize}
  or take a vertex $w_i$ and:
  \begin{itemize}
    \item consider the free composition of the two sp-trees to which $w_i$ belongs
      (joining them at $w_i$), we get a new sp-tree replacing them.
  \end{itemize}
  Repeating this process for all edges of $H_1,\dots,H_y$ and all vertices
  $w_1,\dots,w_z$, we inductively build
  an sp-tree corresponding to $(G,\alpha,\beta)$, for some $\alpha,\beta$.
  Indeed, since $G$ is connected and $H_1,\dots,H_y$ are only 1-connected
  components, all edges of $H_1,\dots,H_y$ will inductively be added to the
  sp-trees, and from the maximality of 2-connected components $G_1,\dots,G_x$
  the free compositions given by $w_i$ always joins two sp-trees that are distinct.
  Consequently this process will eventually merge all ttsp-trees
  into a unique sp-tree corresponding to
  the whole ss-graph $(G,\alpha,\beta)$ (at this point $\alpha,\beta$ could be
  the distinguished vertices of any ttsp-tree among $T_{G_1},\dots,T_{G_x}$,
  and the resulting sp-tree is in free-on-top normal form).

  \medskip

  For the ``if'' part, let $(G,\alpha,\beta)$ be obtained by series, parallel
  and free compositions from base ss-graphs, with $T_G$ a corresponding sp-tree
  in free-on-top normal form. Thanks to the free-on-top normal form, we can adapt
  the tree decomposition presented in~\cite[Lemma 2.3.5]{f97} in order to prove
  that $G$ has treewidth at most 2.
  Let $T^1_G=(N^1,F^1,\sigma^1),\dots,T^x_G=(N^x,F^x,\sigma^x)$ denote the forest
  of ttsp-trees excluding free compositions (but including all base ss-graphs).
  We construct the tree decompositions
  $(X^1,T^1_G),\dots,(X^x,T^x_G)$ for the ttsp-graphs respectively corresponding to
  $T^1_G,\dots,T^x_G$, as follows: $X^i=\{X^i_\eta \mid \eta \in N^i\}$ and
  \begin{itemize}
    \item for each p-node $\eta \in N^i$ of label $(\alpha,\beta)$,
      set $X^i_\eta = \{\alpha,\beta\}$,
    \item for each s-node $\eta \in N^i$ of label $(\alpha,\beta)$ and labels of
      its two children $(\alpha,\gamma)$ and $(\gamma,\beta)$,
      set $X^i_\eta = \{\alpha,\beta,\gamma\}$.
  \end{itemize}
  The tree structure of each tree decomposition is directly given by the ttsp-tree.
  It remains to assemble these tree decompositions according to the free
  compositions, in order to obtain a tree decomposition of $G$. 
  For each f-node of label $(v,v')$, {\em i.e.} identifying vertices $v,v'$
  of its children, one can simply add an edge merging two tree decompositions
  into a single tree decomposition, between any bag containing $v$ and
  any bag containing $v'$ (and rename $v'$ as $v$).
  The result $(X,T)$ of this process is indeed a tree decomposition of $G$:
  \begin{itemize}
    \item the obtained graph $T$ is a tree, as at each step
      two distinct trees are merged,
    \item all base ss-graphs (which are leaf-nodes of $T_G$) belong to
      ttsp-graphs, hence all edges of $G$ are covered in some bag of the tree
      decompositions $(X^1,T^1_G),\dots,(X^x,T^x_G)$,
    \item the subtree associated to each vertex is connected,
      because it was connected in tree decompositions
      $(X^1,T^1_G),\dots,(X^x,T^x_G)$, and the merge of two tree decompositions
      connects identified vertices.
  \end{itemize}
  Since $(X,T)$ has bags of size at most 3,
  it follows that the treewidth of $G$ is at most 2.
  From~\cite{b98,d65,wc83}, the family of graphs of treewidth at most 2
  equals the family of sp-graphs
  (via the equality with {\em $K_4$-minor-free} and {\em partial 2-trees}),
  thus $G$ is an sp-graph.
\end{proof}

\begin{proof}[Proof of Theorem~\ref{theorem:osp}]
  Given a directed graph $G$, one can identify in linear time the
  2-connected components of $\un{G}$ (from \cite{ht73}),
  and check in linear time that each of them is a blind ttsp-graph
  (from~\cite[Section~2.3 and Section~3.3 for implementation details]{vtl79}
  which is based on the characterization as series-parallel
  {\em reducible} multigraphs from~\cite{d65} and its Church-Rosser property,
  with the blind undirected graph oriented as in~\cite[Chapter 3.5]{v78}).
  This last algorithm also builds a ttsp-tree for each 2-connected component.
  In order to get an sp-tree for the whole graph, it remains to include the
  free compositions as in the ``if'' part of the proof of
  Theorem~\ref{theorem:spf}, which is also done in linear time.  We are now
  equipped with an sp-tree (in free-on-top normal form, which is not a
  necessary feature) corresponding to $\un{G}$ (in the case $G$ is an
  osp-graph, otherwise we reject the instance).

  \medskip

  The second (and main) part of the algorithm is a direct application of
  Lemmas~\ref{lemma:seqcomp}, \ref{lemma:parcomp} and~\ref{lemma:freecomp}.
  For a base ss-graph $H$ with an edge oriented as $(u,v)$ we have
  $$
    \begin{array}{c}
      \SUDp{H}{u,v}{+,+}=
      \SUDp{H}{u,v}{-,\varnothing}=
      \SUDp{H}{u,v}{\varnothing,+}=
      \SUDp{H}{u,v}{\varnothing,\varnothing}=
      0\\[.2em]
      \text{and}\quad
      \SUDp{H}{u,v}{+,\varnothing}=
      \SUDp{H}{u,v}{\varnothing,-}=
      1,
    \end{array}
  $$
  then by induction on the structure of the sp-tree of $\un{G}$
  (which provides two distinguished vertices for the osp-graph corresponding
  to each node of the sp-tree) we apply
  Lemma~\ref{lemma:seqcomp} for s-nodes,
  Lemma~\ref{lemma:parcomp} for p-nodes, and
  Lemma~\ref{lemma:freecomp} for f-nodes.
  There are $\O(n \log n)$ steps with $n$ the size of the input graph,
  because the sp-tree has one leaf-node for each arc of $G$ and is a binary tree.
  Finally, all values counting some number of labelings are in $\O(2^n)$
  (the number of $\{\labplus,\labminus\}$-labelings of the whole input
  digraph $G$), hence on $\O(n)$ bits, and the result follows.
\end{proof}

\end{document}